\documentclass{amsart}

\usepackage{amsfonts,mathrsfs}
\usepackage{amsmath,amssymb,amsthm, amscd, bm}
\usepackage{tikz}
\usetikzlibrary{arrows,calc,matrix,topaths,positioning,scopes,shapes,decorations,decorations.markings} 
\usepackage{dsfont}
\usepackage{epsfig}
\usepackage{tikz-cd}
\usepackage{hyperref}
\usepackage{fullpage}
\usepackage[foot]{amsaddr}
\usepackage{adjustbox}
\usepackage{enumitem}

\usepackage[utf8x]{inputenc}

\begin{document}

\theoremstyle{plain}
\newtheorem{theorem}{Theorem}[section]
\newtheorem{main_theorem}[theorem]{Main Theorem}
\newtheorem{proposition}[theorem]{Proposition}
\newtheorem{corollary}[theorem]{Corollary}
\newtheorem{corollaire}[theorem]{Corollaire}
\newtheorem{lemma}[theorem]{Lemma}
\theoremstyle{definition}
\newtheorem{definition}[theorem]{Definition}
\newtheorem{Theorem-Definition}[theorem]{Theorem-Definition}
\theoremstyle{remark}
\newtheorem{remark}{Remark}
\newtheorem{example}{Example}

\sloppy

\newcommand{\Vol}{\operatorname{Vol}}
\newcommand{\RT}{Witten-Reshetikhin-Turaev }
\newcommand{\Homeo}{\operatorname{Homeo}}
\newcommand{\Hom}{\operatorname{Hom}}
\newcommand{\Opp}{\operatorname{Op}}
\newcommand{\GL}{\operatorname{GL}}
\newcommand{\PGL}{\operatorname{PGL}}
\newcommand{\End}{\operatorname{End}}
\newcommand{\Mod}{\operatorname{Mod}}
\newcommand{\U}{\operatorname{U}}
\newcommand{\quotient}[2]{{\raisebox{.2em}{$#1$}\left/\raisebox{-.2em}{$#2$}\right.}}
\newcommand{\sslash}{\mathbin{/\mkern-6mu/}}
\newcommand{\tr}{\operatorname{Tr}}
\newcommand{\PU}{\operatorname{PU}}
\newcommand{\ch}{\operatorname{ch}}
\newcommand{\emcg}{\widetilde{\Mod(\Sigma_g)}}
\newcommand{\mcg}{\Mod(\Sigma_g)}
\newcommand{\SL}{\operatorname{SL}}
\newcommand{\SU}{\operatorname{SU}}
\newcommand{\col}{\operatorname{col}}
\newcommand{\sgn}{\operatorname{sgn}}

\title[A quantum ergodic theorem]{A quantum ergodic theorem for mapping class groups action on character variety}

\author{Julien \textsc{Korinman}}
\address{ Institut Montpelli\'erain Alexander Grothendieck - UMR 5149 Universit\'e de Montpellier. Place Eug\'ene Bataillon, 34090 Montpellier France}
\email{julien.korinman@gmail.com}
\subjclass{$57$R$56$, $81$Q$50$, $57$M$60$}
\keywords{Quantum ergodicity, character variety, \RT representations, mapping class group}
\date{}
\maketitle


\begin{abstract} 
We state a theorem relating the ergodicity of the action of a given subgroup of the mapping class group of a surface on the character variety, to the asymptotic of its invariant subspaces through the \RT representations. As application we give an asymptotic result on the spin decomposition arising in TQFT.
\end{abstract}

\section{Introduction}

\par The purpose of this paper is to generalize a classical quantum ergodic theorem (\cite{Schnirelman, Zelditch, CdV, BdB,Zelditch94}) relating the ergodicity of the action of a symmetry group on a compact phase space to the asymptotic of the decomposition of its associated  group representations arising from quantization.
\vspace{2mm}
\par By a \textit{classical dynamical system} $(\mathcal{A}^{cl}, \tau^{cl}, G)$ we refer to:
\begin{enumerate}
\item a commutative $C^*$ algebra $\mathcal{A}^{cl}$, i.e. a commutative unital $\mathbb{C}$ algebra $\mathcal{A}^{cl}$ with an involutive antilinear, antimorphism of algebra $* : \mathcal{A}^{cl} \to \mathcal{A}^{cl}$ and a norm $\lVert \cdot \rVert$ such that $\lVert f^*f\rVert = \lVert f \rVert^2$; 
\item a state $\tau^{cl}: \mathcal{A}^{cl}\to \mathbb{C}$, i.e. a linear form such that $\tau^{cl}(f^* f)\geq 0$ and $\tau^{cl}(1)=1$; 
\item a group $G$ acting on $\mathcal{A}^{cl}$ by automorphisms of $C^*$-algebra (i.e. commuting with $*$ and preserving the norm) such that $\tau^{cl}$ is $G$ invariant (i.e. $\tau^{cl}(g\cdot f)=\tau^{cl}(f)$ for $g\in G$, $f\in \mathcal{A}^{cl}$).
\end{enumerate}

The main example to keep in mind is the case where $(\mathcal{M},\omega)$ is  a compact symplectic manifold  (phase space) together with a symmetry group $G$ acting on $\mathcal{M}$ by symplectomorphisms. In this case $\mathcal{A}^{cl}=C(M)$ is the algebra of continuous maps $f: \mathcal{M}\to \mathbb{C}$, $*$ is the complex conjugacy, $\lVert f \rVert= \max_{x\in \mathcal{M}} |f(x)|$ and 
$$ \tau^{cl}(f)= \frac{1}{\Vol(\mathcal{M})} \int_{\mathcal{M}} f dV, $$
where $dV$ is the Liouville measure. $G$ acts by $g\cdot f (x) = f(g^{-1}\cdot x)$. A second example is the case where $X$ is the compact real form of an algebraic complex variety $X^{\mathbb{C}}$ (possibly with singularities)  whose smooth locus admits a symplectic form and such that $G$ acts by symplectomorphism; in this case $\mathcal{A}^{cl}$ will be the algebra of regular functions of $X^{\mathbb{C}}$ and $*$ will be the involution defining the compact real form.  Let $L^2(\mathcal{A}^{cl}, \tau^{cl})$ be the Gelfand-Naimark-Segal (GNS) construction (i.e. the completion for $\lVert \cdot \rVert$ of the quotient of $\mathcal{A}^{cl}$ by the kernel of the pairing $(f,g)\mapsto \tau^{cl}(fg^*)$) on which $G$ acts by quotient and completion. The action of $G$ is said $\textit{ergodic}$ if the only $G$ invariant vectors are the scalars $\mathbb{C}\subset L^2(\mathcal{A}^{cl}, \tau^{cl})$. In the previous example, this is equivalent to saying that every $G$-invariant Borel subsets of $\mathcal{M}$ have measure $0$ or $1$  (see \cite{Sunder_VNAlgErgodicity} for details). 
\vspace{2mm}
\par A \textit{quantized dynamical system} $(\mathcal{A}_N, V_N, \Opp_N, \rho_N)_{N\geq N_0}$ consists of:
\begin{enumerate}
\item A family of (non-commutative) $C^*$ algebras $(\mathcal{A}_N)_N$ (quantum observables) thought as a non-commutative deformation of $\mathcal{A}^{cl}$   along a parameter $N$ which plays the role of the inverse of the reduced Planck constant (see Section $2$ for details on quantization).
\item A family of finite dimensional Hilbert spaces  $((V_N, \left<\cdot, \cdot\right>))_N$.
\item Some quantization maps
$$ \Opp_N : \mathcal{A}^{cl} \to \End (V_N) $$
which are morphisms of vector spaces (but not algebras morphisms) and satisfy the positivity condition:
\
$$\Opp_N(f^* f)\geq 0, \mbox{ for all }f\in \mathcal{A}^{cl},$$
where $\Opp_N(f)\geq 0$ means that the operator has non-negative eigenvalues. This condition will be automatically be satisfied if $\Opp_N(f^*)=\Opp_N(f)^*$, i.e. when $\Opp_N(f^*)$ is the dual of $\Opp_N(f)$.

\item A family of unitary representations $(\rho_N)_N$ of a central extension $\hat{G}$ of $G$
$$ \rho_N : \hat{G} \rightarrow \U (V_N),$$
which is related to the quantization through the following asymptotic Egorov identity:
$$ \begin{array}{ll} \left\lVert \rho_N(\phi)^{-1}\Opp_N(f)\rho_N(\phi) -\Opp_N(f\circ \phi)\right\rVert  \xrightarrow[N\to \infty]{} 0 &, \forall f\in \mathcal{A}^{cl} ,\forall \phi \in \hat{G}. \end{array} $$
If moreover the quantization satisfies the equality $ \rho_N(\phi)^{-1}\Opp_N(f)\rho_N(\phi) = \Opp_N(f\circ \phi)$ for all $f$ and $\phi$, we will say that the quantization satisfies the \textit{exact Egorov identity}. 
\end{enumerate}

\vspace{2mm}
\par Given a subspace $W\subset V_N$, one can associate a state $\tau_W$ on $\mathcal{A}^{cl}$ through the formula:
$$ \tau_W(f) := \frac{1}{\dim (W)} \tr \left( \Pi_W \Opp_N (f)\right), $$
 where $\Pi_W$ denotes the orthogonal projection on $W\subset V_N$. When $\mathcal{A}^{cl}=C^{\infty}(\mathcal{M})$, this gives 
 a probability measure $\mu_W$ on $\mathcal{M}$ (equipped with its Borelian $\sigma$-algebra given by $dV$) through the formula $ \int_{\mathcal{M}} f d\mu_W= \tau_W(f)$.
\vspace{2mm}
\par We now state the main theorem of the paper. Suppose that one has a decomposition:
$$ V_N = W_{1,N}\oplus \ldots \oplus W_{p_N, N} $$
 of $V_N$ into $G$-invariant subspaces.
 \vspace{2mm}
 \begin{theorem}\label{main_th}
 Assume that:
 \begin{enumerate}
 \item The group $G$ acts ergodically.
 \item The "quantum average"  of observables converges to the "classical average", \textit{i.e.} the sequence $(\tau_{V_N})_N$ converges, in the $*$-weak topology, to the classical state $\tau^{cl}$. In other words, we ask that for all observables $f\in\mathcal{A}^{cl}$ one has:
 $$\frac{1}{\dim (V_N)} \tr \left(  \Opp_N (f)\right) \xrightarrow[N\to \infty]{} \tau^{cl}(f).$$
 
 \end{enumerate}
 \par Then there exist sets $J_N \subset \{ 1, \ldots, p_N \}$ such that:
 \begin{enumerate}
 \item One has
 $$ \frac{1}{\dim (V_N)} \sum_{i\in J_N} \dim (W_{i,N}) \xrightarrow[N \to \infty]{} 1. $$
 \item For any sequence $j=(j_N)_N$ with $j_N\in J_N$, the sequence $(\tau_{W_{j_N}})_N$ converges in the $*$-weak topology to $\tau^{cl}$. This means that for any classical observable $f\in \mathcal{A}^{cl}$, one has
 $$ \frac{1}{\dim (W_{j_N, N})} \tr \left( \Pi_{W_{j_N,N}} \Opp_N(f) \right) \xrightarrow[N \to \infty]{}\tau^{cl}(f).$$
 \end{enumerate}
 \end{theorem}
 \vspace{2mm}
 \par The conclusion of this theorem should be understood as: " almost every sequence of states $(\tau_{W_{j_N}})_N$ converges  to $\tau^{cl}$".
 \vspace{2mm}
 \par This theorem generalizes previous work (\cite{Schnirelman, Zelditch, CdV, BdB,Zelditch94}) where $G$ was either abelian or amenable. The previous proofs made use of the Birkhoff theorem which only holds for  restricted class of groups (see the introduction of \cite{PollicotSharp} and references therein for a modern discussion on generalizations of the Birkhoff theorem) but does not hold for the more general groups we have in mind, that is the mapping class group of surfaces. Our proof is more elementary and makes use of the fact that the state $\tau_{V_N}$ is the barycenter of the states  $(\tau_{W_{i,N}})_i$ with weights $\alpha_{i,N}:= \frac{\dim (W_{i,N})}{\dim (V_N)}$. The ergodicity of the action of $G$ is equivalent to the fact that the state $\tau^{cl}$ is extremal in the convex compact set $\Delta$ of $G$-invariant states. Theorem \ref{main_th}  will result from the elementary Proposition \ref{prop_geom} which states that if a sequence of finite sets of points in a convex compact metric vector space have barycenters converging to an extremal point, then "almost all" subsequences of its elements converge to the extremal point. Figure $\ref{fig_ergodicity}$ illustrates this proof. As pointed to us by S.Nonnemacher, a similar geometric interpretation already appeared in \cite{Zelditch94}, though it did not lead the author to a geometric proof.  
   
\begin{figure}[!h] 
\centerline{\includegraphics[width=7cm]{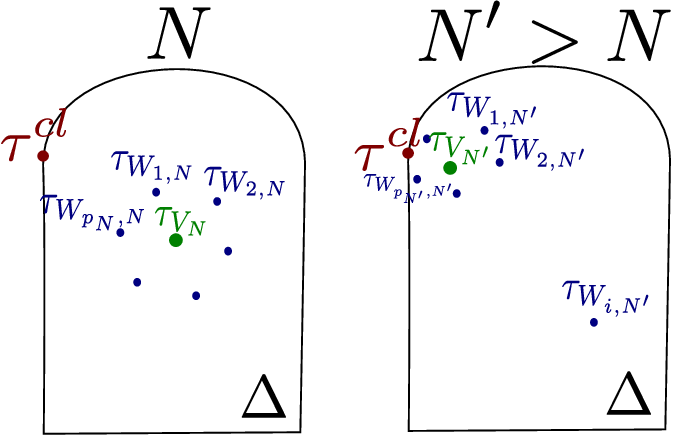} }
\caption{Illustration of Proposition \ref{prop_geom}. A sequence of finite sets of points of a compact convex subset of a vector metric space have barycenters converging to an extremal point. It results that "almost every" subsequences of points converge to the extremal point.} 
\label{fig_ergodicity} 
\end{figure} 

\vspace{2mm}
\par A new feature that does not appear in previous versions of the theorem is that we deal now with invariant spaces of arbitrary dimensions and not just one-dimensional ones. When the dimension of the invariant subspaces are not negligible compare to the dimension of the whole space, we get the following straightforward consequence of Theorem \ref{main_th}:
\begin{corollary}\label{coro_main_th}
 Under the assumptions of Theorem $\ref{main_th}$, if $(j_N)_N$ is an exceptional sequence such that $(\tau_{j_N,N})_N$ does not converge towards $\tau^{cl} $ (in the $*$-weak topology), then one has
 $$ \frac{\dim (W_{j_N,N})}{\dim (V_N)} \xrightarrow[N\to \infty]{} 0. $$ 
 \end{corollary}
 \vspace{2mm}
\begin{remark}  Still it might happen that a few sequences $(\tau_{i_p,p})_p$ do not converge to $\tau^{cl}$ if the associate dimensions are not too large, as illustrated in Figure \ref{fig_ergodicity}. In the case of the quantization of the two-dimensional torus with action given by an Anosov element, such exceptional sequences might converge to measures which are barycenters of extremal points and Dirac measures localized on periodic orbits and are referred as    \textit{Scars} in \cite{FNdB}. Moreover, in \cite{Kelmer}, Kelmer exhibited in the case of higher dimensional tori exceptional sequences converging to measures supported by invariant sub-tori named \textit{super Scars}.
\end{remark}
\vspace{3mm}
\par As main application of our theorem, we apply it to the case familiar to quantum topologists where
$$ \mathcal{M}(\Sigma) := {\Hom\left( \pi_1(\Sigma), SU(2)\right)}\sslash {SU(2)}\quad$$
 is the $\mathrm{SU}(2)$ character variety of a closed oriented connected surface $\Sigma$ equipped with the Atiyah-Bott symplectic form $\omega_{AB}$ (see Section $4$). The mapping class group $\Mod(\Sigma)$ and its first Johnson sub-group $\mathcal{K}(\Sigma)$ naturally act by symplectomorphisms on $(\mathcal{M}(\Sigma), \omega_{AB})$. The classical system $(\mathcal{M}, \omega_{AB}, \Mod(\Sigma))$ admits quantization $\left( \mathcal{S}_{A_N}(\Sigma), V_N(\Sigma), \Opp_N, \rho_N \right)_{N\geq 3}$ first defined heuristically by Witten in \cite{Wi2} and more rigorously by Reshetikhin and Turaev in \cite{RT}. In Section $4$ we briefly review their construction following the skein approach of \cite{Li2, BHMV2}.
 \vspace{2mm}
\par  In \cite{Goldman1}, Goldman showed that the action of $\Mod(\Sigma)$ is ergodic. In \cite{FuMa}, Funar and March\'e showed that the action of the first Johnson subgroup $\mathcal{K}(\Sigma)$ is also ergodic. In \cite{BHMV2}, non trivial invariant subspaces for both $\Mod(\Sigma)$ and $\mathcal{K}(\Sigma)$ where found when $4$ divides $N$. 
\par   Denote by $\mathcal{H}$ the group:
$$ \mathcal{H} := \left\{ \begin{array}{ll}
\mathrm{H}_1 \left( \Sigma_g; \mathbb{Z}/2\mathbb{Z}\right) &\mbox{, if }N \equiv 4 \pmod{8}; \\ 
\mathrm{H}_1 \left( \Sigma_g; \mathbb{Z}/2\mathbb{Z}\right)\rtimes \mathbb{Z}/2\mathbb{Z} &\mbox{, if }N \equiv 0 \pmod{8}. 
\end{array} \right.$$
 where the semi-direct product is the only non trivial one. In \cite{BHMV2}, the authors defined a non trivial decomposition:
 $$ V_N(\Sigma) = \bigoplus_{\chi\in \Hom\left(\mathcal{H},\mathbb{Z}/2\mathbb{Z}\right)} V_N\left(\Sigma,\chi\right),$$
where each $V_N\left(\Sigma,\chi\right)$ is invariant under $\mathcal{K}(\Sigma)$ and $V_N\left(\Sigma,\chi=0 \right)$ is invariant under $\Mod(\Sigma)$. We deduce from Corollary $\ref{coro_main_th}$  the following:
\begin{corollary}\label{coro_spin} {}
\par Let $(N_k)_k$ be an increasing sequence of non-negative integers, all of which being congruent to either $4$ or $0$ modulo $8$, and let $\chi \in \Hom\left(\mathcal{H},\mathbb{Z}/2\mathbb{Z}\right)$. Denote by $\Pi_{\chi, N}$ the orthogonal projector on $V_N\left(\Sigma,\chi\right)$. Then, for all $f\in \mathcal{O}_{\mathcal{M}(\Sigma)}$, one has
$$ \frac{1}{\dim (V_{N_k}\left(\Sigma,\chi\right))} \tr \left( \Opp_{N_k}(f)\Pi_{\chi,N_k}  \right) \xrightarrow[k \to \infty]{} \frac{1}{\Vol (\mathcal{M}(\Sigma))} \int_{\mathcal{M}(\Sigma)} f dV.$$
\end{corollary}
\vspace{2mm}
\par The original Schnirelman theorem, proved in \cite{Schnirelman, Zelditch, CdV}, does not immediately follow from Theorem \ref{main_th} because the Hilbert space considered is infinite dimensional. However it easily follows from Proposition \ref{prop_geom} as follows.
\vspace{2mm} \par
 In this case the classical system is $\left( S^* M, \omega, \mathbb{R} \right)$ where $(M,g)$ is a compact Riemannian hyperbolic manifold geodesically complete, $S^*M$ denotes the unitary cotangent bundle, $\omega$ the canonical symplectic form, the action of $\mathbb{R}$ is the geodesic flow and the algebra of classical observables is $C^{\infty}(S^*M)$. 
\vspace{2mm} \par 
The associated quantum system is $(\mathcal{H}, {\Psi}^0 (M), \Opp^F, U)$, where $\mathcal{H}=L^2(M)$ and ${\Psi}^0(M)$ denotes the algebra of order $0$ pseudo-differential operators on $M$.  The map $\Opp^F : C^{\infty}(M) \rightarrow {\Psi}^0(M)$ is Friedrich's quantization map (see \cite{CdV}) which satisfies the positivity condition $f\geq 0 \implies \Opp^F(f)\geq 0$ and whose inverse is the principal symbol map. The unitary representation of $\mathbb{R}$ on $\mathcal{H}$ is given by: 
$$ U^t \psi := \exp\left(i \sqrt{\Delta_M}t\right)\psi, \quad t\in \mathbb{R}.$$
\par This quantization satisfies the  asymptotic Egorov identity.
\vspace{2mm}
\par When $(M,g)$ is hyperbolic and geodesically complete, the geodesic flow induces an ergodic action of $\mathbb{R}$ on $S^*(M)$ hence the first condition of Theorem \ref{main_th} is satisfied. The second condition takes the following form. Let $(\psi_k)_k$ be a sequence of norm one eigenvectors of the Laplacian, that is $\Delta_M \psi_k = \lambda_k \psi_k$, indexed such that $(\lambda_k)$ is an increasing sequence of eigenvalues. Let $\tau_k$ denotes the associated state, that is
$$ \tau_k(f):= \left< \psi_k, \Opp^F(f)\psi_k\right>$$
and denotes by $\tau^{cl}$ the state associated to the Liouville measure of $\omega$. For $N\geq 0$ define the barycenter state
$$ B_N:= \frac{1}{\# \{ k, \lambda_k\leq N\} } \sum_{\lambda_k\leq N} \tau_k.$$
\par An application of Kamarata's Tauberian theorem (see \cite[Paragraph $4$]{CdV}) gives
$$ B_N \xrightarrow[N \to \infty]{}  \tau^{cl}$$ 
which is the analogue of the second condition of Theorem \ref{main_th}. Now Proposition \ref{prop_geom} implies that almost every sub-sequences $(\tau_{N_p})_p$ converges to $\tau^{cl}$:  this is the classical  Schnirelman theorem.

\vspace{3mm}
\par The paper is organized as follows. In Section $2$ we review the notion of quantization and detail the familiar example 
of the Schr\"odinger quantization of the two-dimensional torus which gives rise to the Weil representation of $\mathrm{SL}_2(\mathbb{Z})$. The image of Anosov elements are usually referred as "Arnold's quantum cats maps" (\cite{HB}) and are the object of study of a previous version of our theorem in \cite{BdB}. Section $3$ is devoted to the proof of Theorem $\ref{main_th}$. In Section $4$ we briefly review the quantization of the $\mathrm{SU}(2)$ character variety and prove Corollary $\ref{coro_spin}$.

\vspace{4mm}
 $\textbf{Acknowledgements:}$ The author is thankful to L.Charles,  L.Funar, J.March\'e, S.Nonnenmacher and F.Paulin for useful discussions. He also warmly thanks  L.Benard, P.Roche and N.Rougerie for useful comments which improved the clarity of the paper. He eventually thanks C.Oliveira, A.E.Presotto, F.Ruffino, D.Vendruscolo and the mathematic department of UFSCar for their kind hospitality during the redaction of this paper. He acknowledges  support from the grant ANR $2011$ BS $01 020 01$ ModGroup, CAPES, the GDR Tresses, the GDR Platon,  the GEAR Network and the  European Research Council (ERC DerSympApp) under the European Union’s Horizon 2020 research and innovation program (Grant Agreement No. 768679).

\section{Quantization of classical systems}

\subsection{Quantum system}
\par To motivate the physical meaning of this paper, we provide a recipe to construct quantizations of a classical system $(\mathcal{A}^{cl}, \tau^{cl}, G)$. Suppose that $\mathcal{A}^{cl}$ has a Poisson bracket $\{\cdot, \cdot\}$ (the one given by the symplectic structure when $\mathcal{A}^{cl}=C(\mathcal{M})$).
\vspace{2mm}
\par Let $\mathbb{C}[[\hbar]]$ denote the field of formal series in some parameter $\hbar$ referred as the reduced Planck constant. Set
 $\mathcal{A}^{cl} [[\hbar ]]:= \mathcal{A}^{cl}\otimes_{\mathbb{C}}\mathbb{C}[[\hbar]] $ 
 seen as a $\mathbb{C}[[\hbar]]$ flat module. We thus consider formal series of functions. A star-product on $\mathcal{A}^{cl} [[\hbar ]]$ is an associative product $\star$ such that if $f$ and $g$ are elements of $\mathcal{A}^{cl} [[\hbar ]]$ with expansion
$$ \begin{array}{ll} f=f_0 + \hbar f_1 + \ldots &, g=g_0+\hbar g_1+\ldots, \end{array}$$
then one has 
$$ f\star g = f_0g_0 + \circ(\hbar)$$
and
$$ f\star g-g\star f = \frac{\hbar}{i} \left\{ f_0, g_0 \right\} + \circ(\hbar^2)$$
 where $\left\{ \cdot, \cdot \right\}$ stands for the Poisson bracket. We refer to (\cite{KontsevichQuantizationPoisson}, \cite[II.2]{GRS_QuantizationDeformation}) for more details on quantization deformation. 
The algebra $\mathcal{A}^{cl} [[\hbar ]]$ with product $\star$ is thus a non-commutative deformation of the algebra of regular functions whose first order expansion is given by the symplectic structure of the phase space.  Consider the values $\hbar=\frac{1}{N}$ where $N$ denotes a positive integer and the complex  vector spaces
 $$\mathcal{A}_N:= \mathcal{A}^{cl} [[\hbar ]]\otimes_{\hbar=\frac{1}{N}} \mathbb{C}.$$ 
 with product given by the star-product. Note that $\mathcal{A}^{cl}$ and $\mathcal{A}_N$ are canonically isomorphic as vector spaces.
 We make the strong assumption that the star-product induces a well defined product on $\mathcal{A}_N$. This assumption will be satisfied if there exists an algebra $\mathcal{A}_q$ flat over $\mathbb{C}[q, q^{-1}]$ such that $\mathcal{A}^{cl}=\mathcal{A}_q\otimes_{q=1}\mathbb{C}$ is obtained by replacing $q$ by $+1$ and $\mathcal{A}^{cl}_{\hbar}\cong \mathcal{A}_q\otimes_{q= \exp(\hbar)} \mathbb{C}[[\hbar]]$ is obtained by replacing $q$ by $\exp(2i\pi \hbar)$. In this case, $\mathcal{A}_N$ is the algebra obtained by replacing $q$ by the root of unity $\exp(2i\pi/N)$. This class of examples includes quantum tori, quantum enveloping algebras, quantum groups, skein and stated skein algebras.
\par 
  A \textit{quantization}  is then given by the star-product together with  representations $\pi_N: \mathcal{A}_N\to \End(V_N)$ which we will assume to be finite dimensional in this paper and equipped with a definite positive Hermitian form $\left<\cdot, \cdot\right>$.  We then define the linear map
$$ \Opp_N: \mathcal{A}^{cl} \cong \mathcal{A}_N  \xrightarrow{\pi_N}  \End(V_N).$$

\par We will make the assumption that $\Opp_N(f^*f)$ has non-negative spectra for every $f$ (positive quantization). This will be automatically satisfied if $\Opp_N(f^*)$ is the dual (for $\left<\cdot, \cdot\right>$) of $\Opp_N(f)$; in this case we will say that $\pi_N$ is a $*$-representation. 
\vspace{2mm}
\par The action of $G$ on $\mathcal{A}^{cl}$ induces an action of $G$ on $\pi_N(\mathcal{A}_N)\subset \End(V_N)$. We will assume that $G$ acts by inner unitary automorphisms, which means that we have a projective representation
$$ \rho_N: G \rightarrow \PU(V_N)$$
such that the following Egorov identity holds
$$ \begin{array}{ll} \left\lVert \rho_N(\phi)^{-1}\Opp_N(f)\rho_N(\phi) -\Opp_N(f\circ \phi)\right\rVert  \xrightarrow[N\to \infty]{} 0 &, \forall f\in \mathcal{A}^{cl} ,\forall \phi \in G. \end{array} $$
This condition will be automatically satisfied if $\pi_N$ is an irreducible representation since, in that case, $\pi_N(\mathcal{A}_N)=\End(V_N)$ by the Schur lemma and all automorphisms of $\End(V_N)$ are inner; so $\rho_N$ exists and unique in this case.
 To deal with linear representations rather than projective ones, we then choose a central extension $\hat{G}$ of $G$ and a lift, still denoted $\rho_N$, to a linear unitary representation
$$ \rho_N: \hat{G} \rightarrow \U(V_N).$$

\par Eventually, the data $\left(\mathcal{A}_N, V_N, \Opp_N, \rho_N\right)_N$ will be referred to as a \textit{quantum system} associated to the classical one $(\mathcal{A}^{cl}, \tau^{cl}, G)$.

\subsection{Example of the two-dimensional torus}

\par The most studied example of such quantized system is the quantization of the two-dimensional torus $\mathbb{T}^2=U(1)^2\cong \quotient{\mathbb{R}^2}{\mathbb{Z}^2}$  with symplectic structure induced from the symplectic form $\omega=\begin{pmatrix} 0&-1\\ 1&0 \end{pmatrix}$ on $\mathbb{R}^2$. The symplectic action of the group $G=\mathrm{SL}_2(\mathbb{Z})$ on $\mathbb{R}^2$ passes to the quotient by $\mathbb{Z}^2$ giving a symplectic action on $\mathbb{T}^2$. We refer to \cite{EspostiGraffi, LV, GU, KR, Koju1} for various equivalent descriptions of its quantization and construction of the Weil representations, which we now summarize.
\vspace{2mm}
\par Consider the complex torus $\mathbb{T}^2_{\mathbb{C}}=(\mathbb{C}^*)^2$ whose algebra of regular functions is $\mathcal{O}_{\mathbb{T}^2}=\mathbb{C}[x^{\pm 1}, y^{\pm 1}]$ (here $x,y: \mathbb{T}^2_{\mathbb{C}}\to \mathbb{C}$ are the functions $x((z_1,z_2))=z_1$ and $y((z_1,z_2))=z_2$).
Then $\mathbb{T}^2$ is the compact real form of $\mathbb{T}^2_{\mathbb{C}}$ associated to the $*$ involution defined by $x^*:= x^{-1}$ and $y^*=y^{-1}$. Said differently, a point $(z_1,z_2)\in \mathbb{T}^2$, corresponding to the character $\chi_{(z_1,z_2)}: \mathcal{O}_{\mathbb{T}^2} \to \mathbb{C}$ sending $x$ and $y$ to $z_1$ and $z_2$, belongs to $\mathbb{T}^2= U(1)^ 2 \subset (\mathbb{C}^*)^2$ if and only if $\chi_{(z_1,z_2)}(f^*)= \overline{\chi_{(z_1,z_2)}(f)}$ for all $f\in \mathcal{A}^{cl}$ (i.e. if and only if $z_1=\overline{z_1}^{-1}$ and $z_2=\overline{z_2}^{-1}$). 
We set $\lVert f \rVert = \max_{(e^{i\theta_1}, e^{i\theta_2})\in \mathbb{T}^2} |f(e^{i\theta_1}, e^{i\theta_2})|$. So the commutative $C^*$ algebra $(\mathbb{C}[x^{\pm 1}, y^{\pm 1}], *, \lVert\cdot \rVert)$ describes our phase space $\mathbb{T}^2$. The Liouville measure associated to the symplectic form $\omega$ induces the classical state $\tau^{cl}(f):= \frac{1}{Vol(\mathbb{T}^2)}\int_{\mathbb{T}^2} f dV$ characterized by $\tau^{cl}(x^ny^m)= 0$ if $(n,m)\neq (0,0)$ and $\tau^{cl}(1)=1$. 
\vspace{2mm}
\par A star-product on $\mathbb{C}[x^{\pm 1},y^{\pm 1}][[\hbar]]$ is then given by the formula
$$ x^ay^b \star x^{a'}y^{b'} := \exp\left( \pi i(b'a-a'b)\hbar\right) x^{a+a'}y^{b+b'}. $$
The algebras $\mathcal{A}_N:=
\left\{ \begin{array}{ll}
 \mathbb{C}[x^{\pm 1},y^{\pm 1}][[\hbar]]\otimes_{\hbar=2/N} \mathbb{C}& \mbox{, if }N\mbox{ is odd,}\\
\mathbb{C}[x^{\pm 1},y^{\pm 1}][[\hbar]]\otimes_{\hbar=1/N} \mathbb{C}& \mbox{, if }N\mbox{ is even.}
\end{array}\right. $ are naturally presented by generators $x^{\pm 1}$ and $y^{\pm 1}$ and relations
$$ xy=A^2 yx$$
where $A=\left\{ \begin{array}{ll}
\exp\left( \frac{2i\pi}{N} \right) & \mbox{, if }N\mbox{ is odd,}\\
\exp\left( \frac{i\pi}{N} \right)& \mbox{, if }N\mbox{ is even.}
\end{array}\right. $  The algebras $\mathcal{A}_N$ are usually referred to as \textit{quantum tori} in literature whereas the groups generated by the elements $x^{\pm 1}$, $y^{\pm 1}$ and $A^2$ are called the \textit{Heisenberg groups}.
\vspace{2mm}
\par Irreducible representations $\pi_N: \mathcal{A}_N\rightarrow \End(V_N)$ are defined by setting $V_N:= \mathbb{C}^N$ Hermitian with orthonormal basis $(e_1, \ldots, e_N)$ and by the following formulas
$$ \begin{array}{ll} \pi_N(x) e_i:= A^{2i} e_i & \pi_N(y)e_i:= e_{i+1} \end{array}$$
 where indexes are taken modulo $N$.  They are referred as the \textit{Schr\"odinger representations}.  We can now define $\Opp_N$ via $\Opp_N(x^n y^m):= \pi_N( A^{-nm}\pi_N(x^ny^m))$. The representations $\pi_N$ are easily showed to be irreducible $*$ representations (since if $M$ represents the matrix of either  $\pi_N(x)$ or $\pi_N(y)$ in the basis $(e_i)_i$, then $M^{-1}={}^t\overline{M}$)
 and to satisfy the quantum average condition of Theorem $\ref{main_th}$. Some authors  extend by density this quantization from the algebra of regular functions to the algebra of smooth functions (see \textit{e.g.} \cite{KR}).
\vspace{2mm}
\par  Now projective representations $\rho_N : \mathrm{SL}_2(\mathbb{Z})\rightarrow \PU (V_N)$ are defined on the generators $T=\begin{pmatrix} 1 & -1 \\ 0 & 1\end{pmatrix}$ and $S=\begin{pmatrix} 0 & 1 \\ -1 & 0 \end{pmatrix}$ by the formulas
$$ 
 \rho_N(S)e_i =  \left(\frac{1}{N}\sum_{k\in\mathbb{Z}/N\mathbb{Z}}A^{k^2} \right) \sum_j A^{2ij}e_j; \quad \rho_N(T) e_i = A^{i^2} e_i. $$

These projective representations can be lifted to linear representations of $\mathrm{SL}_2(\mathbb{Z})$. They were first defined by Kloosterman in \cite{Kl}, and usually referred to as \textit{Weil representations}. They satisfy the exact Egorov identity. This fact can be compared to the case of higher dimensional tori where the equivalent projective representations of the symplectic groups have to be centrally extended by $\mathbb{Z}/2\mathbb{Z}$ to be lifted to linear ones when $N$ is even. It is however more usual to rather centrally extend the symplectic group by $\mathbb{Z}$ thus obtaining the so-called \textit{metaplectic group} to lift the Weil representations. In many textbook, when $N$ is odd, the Weil representation is only defined on a sub-group of the symplectic group due to its first definitions related to modular forms and Theta functions. 
\vspace{2mm}
\par Eventually let $\phi\in \mathrm{SL}_2(\mathbb{Z})$ be an Anosov element, that is a matrix such that $|\tr(\phi)|>2$. The action of $G=\mathbb{Z}$ on $\mathbb{T}^2$ where $n$ acts by $\phi^n$, is well known to be ergodic. The quantum system $\left(\mathcal{A}_N, V_N, \Opp_N, \rho_N\right)$ is a quantization of the classical one $(\mathcal{O}_{\mathbb{T}^2_{\mathbb{C}}}, \tau^{cl}, \mathbb{Z})$  and the operators $\rho_N(\phi)$ are called \textit{Arnold's quantum cat map} (see \cite{HB, BdB, KR}). Let $(v_{1,N}, \ldots, v_{N,N})$ be a basis of $V_N$ consisting of norm one eigenvectors of $\rho_N(\phi)$ such that one has the following decomposition:
$$V_N = \mathbb{C}v_{1,N}\oplus \ldots \oplus \mathbb{C}v_{N,N} $$
 of $V_N$ into one-dimensional  sub-spaces invariant for the action of $\mathbb{Z}$.  Theorem $\ref{main_th}$ states that for almost every sequence $j= (j_N)_N$ and all polynomials $f\in \mathbb{C}[x^{\pm 1},y^{\pm 1}]$, one has 
 $$ \left< v_{j_N}, \Opp_N(f) v_{j_N}\right>_N \xrightarrow[N \to \infty]{}  \frac{1}{\Vol (\mathbb{T})^2} \int_{\mathbb{T}^2} f dV.$$
 \par This is the main theorem of \cite{BdB}. Our theorem is thus a generalization of this quantum ergodic theorem in the case where $G$ is a general group.

\section{The quantum ergodic theorem}

\subsection{States associated to invariant subspaces in positive quantization}
\par Recall that a state on the $C^*$ algebra $\mathcal{A}^{cl}$ is a continuous $*$-linear form $\tau : \mathcal{A}^{cl} \to \mathbb{C}$ such that $\tau(1)=1$ and $\tau(f^*f)\geq 0$. $\tau$ is $G$ invariant if $\tau(g\cdot f)=\tau(f)$ for all $f\in \mathcal{A}^{cl}$ and $g\in G$. When $\mathcal{A}^{cl}=C(\mathcal{M})$, there is a bijection between the set of probability measures on $\mathcal{M}$ and the set of states  sending the measure $\mu$ to the state $\tau_{\mu}$ defined by
$$ \tau_{\mu}(f):= \frac{1}{\mu(\mathcal{M})}\int_{\mathcal{M}} f d\mu.$$

\begin{definition} Let $W\subset V_N$ be a sub-space. Let $\Pi_W$ be the orthogonal projection on $W$. The state $\tau_W$ is defined by
$$ \tau_W(f):= \frac{1}{\dim(W)} \tr\left(\Pi_W \Opp_N(f)\right).$$
\end{definition}
\par The positivity of $\tau_W$ results from the positivity of the quantization.

\par Let $(E,d)$ denotes the dual space of $(\mathcal{A}^{cl}, \lVert \cdot \rVert)$ equipped with a metric of the $*$-weak topology. For instance, if $(f_n)_n$ is a countable family  dense in $(\mathcal{A}^{cl}, \lVert \cdot \rVert)$, we can choose
$$ d(\tau_1, \tau_2):= \sum_n \frac{1}{2^n}\lvert \tau_1(f_n)-\tau_2(f_n)\rvert. $$

The set $\Delta$ of $G$-invariant states is a convex compact subset of the (compact) unit ball of $(E,d)$. 
\par Recall that $G$ acts ergodically on $(\mathcal{A}^{cl}, \tau^{cl})$ if the only $G$ invariant vectors of $L^2(\mathcal{A}^{cl}, \tau^{cl})$ are scalars. When $\mathcal{A}^{cl}=C(\mathcal{M})$ and $\tau^{cl}$ corresponds to the Liouville measure, this equivalent to the fact that 
 every $G$-invariant Borel set has either measure $0$ or $1$. The following lemma is classical:
 
 \begin{lemma}
$G$ acts ergodically on $(\mathcal{A}^{cl},   \tau^{cl})$ if and only if $\tau^{cl}$ is an extremal point of $\Delta$.
\end{lemma}

\subsection{Proof of Theorem $\ref{main_th}$}

\par Under the assumptions of Theorem $\ref{main_th}$, the state $\tau_{V_N}$ is the barycenter of the states $(\tau_{ W_{1,N}}, \ldots, \tau_{W_{p_N,N}})$ with weights $\alpha_{i,N}:= \frac{\dim (W_{i,N})}{\dim (V_N)}$ in the vector metric space $(E,d)$, that is we have $\sum_i \alpha_{i,N} = 1$ and $\tau_{V_N}= \sum_i \alpha_{i, N} \tau_{W_{i, N}}$. The hypotheses of Theorem $\ref{main_th}$ imply that the sequence of the barycenters $\tau_{V_N}$ converges to the extremal point $\tau^{classical}$ in the convex compact $\Delta$. Theorem $\ref{main_th}$ thus results from the following proposition, whereas Corollary $\ref{coro_main_th}$ is an easy consequence:

\begin{proposition}\label{prop_geom}.
 Let $(E,d)$ be a metric vector space and $\Delta\subset E$ be a convex compact subset. 
\vspace{2mm}
\par For all $N\geq 1$, fix an integer $p_N>0$ and some points $\tau_{1,N}, \ldots, \tau_{p_N,N} \in \Delta$ together with weights $\alpha_{1,N},\ldots, \alpha_{p_N,N} \in [0,1]$ such that  $\sum_i \alpha_{i,N}=1$.
\vspace{2mm}
\par Denote by $B_N:=\sum_i \alpha_{i,N} \tau_{i,N}$ the barycenter of the weighted points. Eventually choose $\tau \in \Delta$ an extremal point of $\Delta$.
\vspace{2mm}
\par Suppose that
$$ B_N \xrightarrow[N \to \infty]{d} \tau.$$
 Then there exist subsets $J_N \subset \{1, \ldots , p_N\}$ such that:
\begin{enumerate}
\item If we note $\|J_N\|:=\sum_{i\in J_N}\alpha_{i,N}$, then
$$ \|J_N\|\xrightarrow[N \to \infty]{} 1.$$
\item For any sequence $j=(j_N)_{N\geq 1}$ with $j_N \in J_N$, one has
$$ \tau_{j_N,N} \xrightarrow[N \to \infty]{d} \tau.$$
\end{enumerate}
\end{proposition}
\vspace{2mm}
\par Figure $\ref{fig_ergodicity}$ illustrates the proposition by showing two sets of points inside a compact convex at two different moments. When the barycenter approaches an extremal point, then 'almost all' points approach it as well. This is our geometric interpretation of the Schnirelman theorem.

\vspace{2mm}
\par The proof of Proposition $\ref{prop_geom}$ will be deduced from the following:

\begin{lemma}\label{lemma_geom}
Let $(E,d)$ be a metric vector space, $\tau_1,\ldots, \tau_p \in E$ be some points equipped with weights $\alpha_1,\ldots,\alpha_p\in [0,1]$ such that $\sum_i \alpha_i =1$.
and denote by $B:=\sum_i \alpha_i \tau_i$ their barycenter. Let $\tau \in E$ be a point such that  $\tau$ is an extremal point of the convex hull $\widetilde{\Delta} := \ch \left(\tau,\tau_1,\ldots, \tau_p \right)$ of the points $\tau$ and $\tau_i$.
\vspace{2mm} \par
 For a subset $K\subset E$, we will use the notation $\lvert K \rvert:= \sum_{i , \tau_i \in K} \alpha_i \in [0,1]$, that is the sum of the weights of every points $\tau_i$ which are inside $K$.
\vspace{2mm}
\par Then for all $\epsilon >0$ and for all $0<\delta <1$, there exists $d>0$ such that
$$ d(\tau, B)\leq d \implies \left| B(\tau, \epsilon) \right| \geq \delta, $$
where $B(\tau, \epsilon)$ denotes the ball of center $\tau$ and  radius $\epsilon$. 
\end{lemma}

\begin{remark}
In this lemma, instead of having a general convex compact as in Proposition $\ref{prop_geom}$, we choose the convex hull of the points $\tau_i$ and $\tau$. The reason for this choice is that in order to provide the linear map $L$ appearing in the following proof, we need the extremal point $\tau$ to be \textit{exposed}.
The point $\tau^{cl}$ of Figure \ref{fig_ergodicity} is an example of a not exposed extremal point.
\end{remark}

\begin{proof}[Proof of $\ref{lemma_geom}$] {}
\par First, it is a general fact (see \textit{e.g.} \cite{Brezis}) that in any locally compact space $E$, we can find a continuous linear form $L\in E'$ separating finite sets of points, that is such that there exists $c\in \mathbb{R}$ such that
$$\begin{array}{ll} L(\tau)<c<L(\tau_i) &\mbox{, for all }i\in \{1,\ldots, p\}. \end{array}$$
  We define $a:=L(\tau)$ and $b:=\max_i \{L(\tau_i)\}$, hence $L(\widetilde{\Delta})=[a,b]$.
 Next fix $\epsilon >0$ and $0<\delta<1$ and choose $c\in (a,b)$ such that
$$ F_c := \left\{ x \in \widetilde{\Delta}\mbox{, s.t }L(x)\leq c \right\} \subset B(\tau, \epsilon). $$
 Writing $\alpha := |F_c|$,  we now want to show that 
$$ d(\tau, B) > \frac{(1-\alpha)c-a}{\Vert L \Vert}.$$
Define the sub-barycenters $B_1:=\sum_{i | \tau_i \in F_c}\alpha_i$ and $B_2:= \sum_{i | \tau_i\notin F_c} \alpha_i$. 
\\ By convexity of $\widetilde{\Delta} \backslash F_c$, one has $L(B_2)>c$. Moreover, since
$$ B= \alpha B_1 + (1-\alpha)B_2, $$
  we have
 $$ L(B)=\alpha L(B_1) +(1-\alpha)L(B_2).$$
 Using the fact that $L$ is continuous, we obtain
 \begin{eqnarray*}
  d(\tau, B) &\geq & \frac{\left| L(B)-L(\tau)\right|}{ \Vert L \Vert} = \frac{ \alpha L(B_1)+(1-\alpha )L(B_2) - a}{\Vert L \Vert } 
  \\ & \geq & \frac{(1-\alpha)L(B_2) - a}{\Vert L \Vert } 
  \\ &>& \frac{(1-\alpha)c-a}{\Vert L \Vert}
  \end{eqnarray*}
 Figure \ref{fig_lemma_geom} illustrates the proof.
     
\begin{figure}[!h] 
\centerline{\includegraphics[width=6cm]{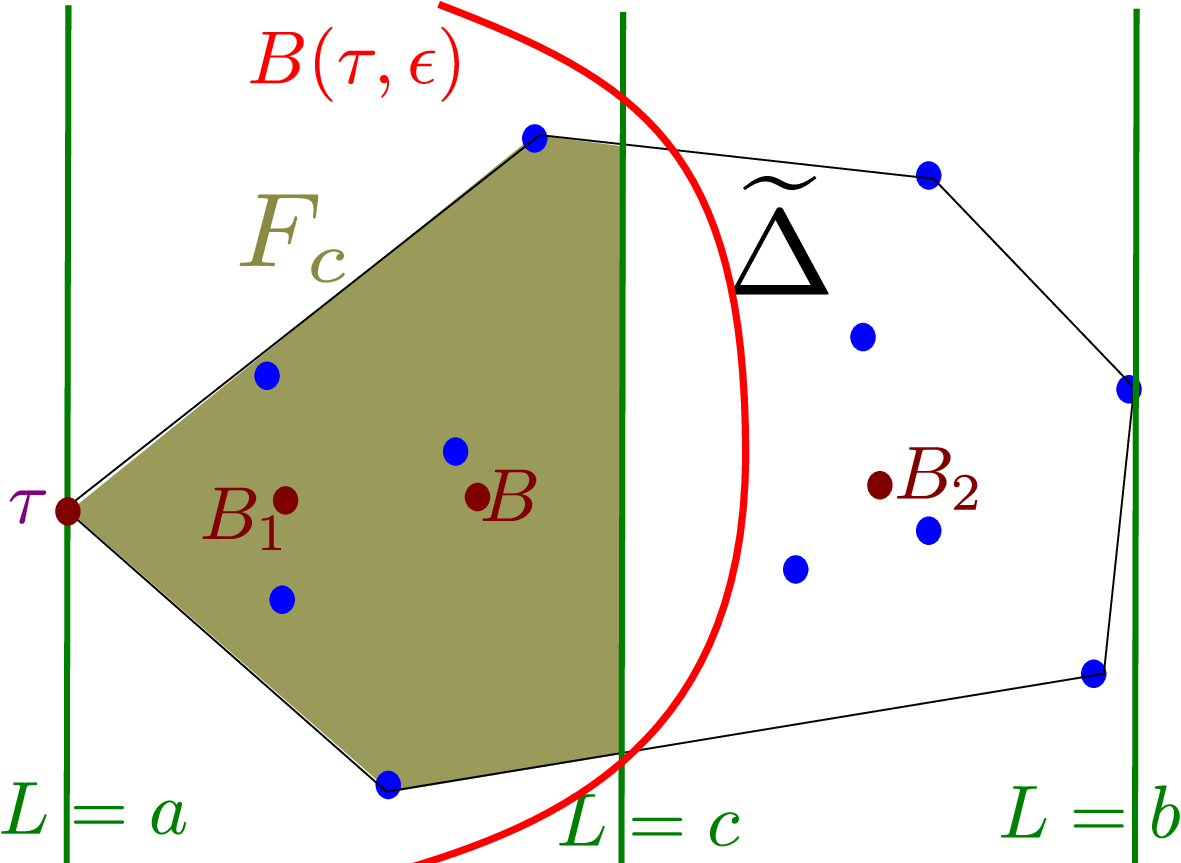} }
\caption{Illustration of the proof of Lemma \ref{lemma_geom}.} 
\label{fig_lemma_geom} 
\end{figure} 
  \vspace{2mm}
  \par To conclude the proof, we remark that since $B(\tau,\epsilon)\subset F_c$, we have $\left| B(\tau,\epsilon)\right| \leq |F_c| = \alpha$. Therefore, one has
  $$ d(\tau, B) > \frac{(1-\left| B(\tau, \epsilon)\right|)c-a}{\Vert L \Vert},$$
  and $d:=\frac{(1-\delta)c-a}{\Vert L \Vert}$ verifies the conclusion of the lemma.
  \end{proof}

  \begin{proof}[Proof of Proposition $\ref{prop_geom}$] .
  \par Fix $r>0$ and apply Lemma $\ref{lemma_geom}$ to the polytope $\widetilde{\Delta_r} := \ch \left\{ \tau_{1,r}, \ldots, \tau_{p_r,r}, \tau \right\}\subset \Delta$ with $\epsilon= \frac{1}{r}$ and $\delta=1-\frac{1}{r}$. We obtain $d_r>0$ such that
  $$ d(\tau, B_r)\leq d_r \implies \left| B\left(\tau, \frac{1}{r}\right) \right| \geq 1-\frac{1}{r}. $$
  Since $(B_N)_N$ converges to $\tau$, there exists a rank $N(r)$ such that
  \begin{eqnarray*} N\geq N(r) &\implies & d(\tau, B_N) \leq d_r \\
   & \implies & \left| B\left(\tau, \frac{1}{r}\right) \right| \geq 1-\frac{1}{r}.
   \end{eqnarray*}
    We can suppose that the sequence $(N(r))_r$ is strictly increasing. For each $N>0$, there exists a unique $r>0$ such that  $N(r)\leq N < N(r+1)$. We set
   $$ J_N:= \left\{ j \mbox{, s.t. }d(\tau, \tau_{j,N}) \leq \frac{1}{r} \right\}. $$
   \begin{enumerate}
   \item Since $N\geq N(r)$, we have $|J_N| = \left| B(\tau, \frac{1}{r})\right| \geq 1-\frac{1}{r}$. Therefore, one has
   $$ | J_N | \xrightarrow[N\to \infty]{} 1.$$
   \item For any sequence $j=(j_N)_N$ with $j_N\in J_N$, one has $d(\tau, \tau_{j_N, N})\leq \frac{1}{r}$. We deduce the convergence
   $$ d(\tau, \tau_{j_N, N})\xrightarrow[N\to \infty]{} 0.$$
   \end{enumerate}
   \end{proof}

\section{The $\SU(2)$ character variety and the \RT representations}

\subsection{Character varieties}

\par Our main application of Theorem $\ref{main_th}$ concerns the \RT quantization of the $\SU(2)$ character variety. Given a closed oriented connected surface $\Sigma$, its associated character variety is the affine variety:
$$ \mathcal{M}(\Sigma) := {\Hom\left( \pi_1(\Sigma), \SU(2)\right)}\sslash {\SU(2)}.$$
It is a compact real form of the $\SL_2$ character variety: 
$$ \mathcal{M}(\Sigma)_{\mathbb{C}} := {\Hom\left( \pi_1(\Sigma), \SL_2(\mathbb{C}) \right)}\sslash {\SL_2(\mathbb{C})}.$$
Said differently, the space of representations $\mathcal{R}(\Sigma):=\Hom\left(\pi_1(\Sigma), \SL_2(\mathbb{C})\right)$ as a natural structure of affine variety on which the group $\SL_2(\mathbb{C})$ acts algebraically by conjugacy. The algebraic quotient we consider (whence the notation with a double bar for the quotient) is by taking the maximal spectrum  of the sub-algebra of  $\SL_2(\mathbb{C})$ invariant functions of $\mathcal{O}_{\mathcal{R}(\Sigma)}$. This quotient, familiar in Geometric Invariant Theory, can be thought as the smallest Haussdorf quotient possible (see \textit{e.g.} \cite{Sikora} for details). 

 By a \textit{multicurve} in $\Sigma$, we mean an isotopy class of embedded compact one-manifold in $\Sigma$ without contractible component. To any element $\gamma=\gamma_1\bigsqcup\ldots \bigsqcup \gamma_n$ of the set $\mathcal{MC}$ of multicurves, we associate a regular function $f_{\gamma}\in \mathcal{O}_{\mathcal{M}(\Sigma)_{\mathbb{C}}}$ by the formula:
 $$ f_{\gamma}([\rho]):= \prod_i \left( -\tr \rho(\gamma_i)\right).$$
 \par In \cite{Bullock,Sikora} it was proved that the set $\{ f_{\gamma}, \gamma\in \mathcal{MC}\}$ forms a basis of the algebra $\mathcal{O}_{\mathcal{M}(\Sigma)_{\mathbb{C}}}$ (see also \cite{ChaMa} for an alternative proof).
The $\SU(2)$-character variety is the compact real form of the $\SL_2(\mathbb{C})$ character variety associated to the $*$ involution defined by $f_{\gamma}^*= f_{\gamma}$. 
For $f \in \mathcal{O}_{\mathcal{M}(\Sigma)_{\mathbb{C}}}$, we set $\lVert f \rVert := \max_{[\rho]\in \mathcal{M}(\Sigma)} f([\rho])$. 
Both $\mathcal{M}(\Sigma)$ and $\mathcal{M}(\Sigma)_{\mathbb{C}}$ are singular, their smooth loci are the loci of classes $[\rho]$ of irreducible representations $\rho$ when $\Sigma$ has genus $g\geq 2$ and is the locus of non central representations when $g=1$. 
 \vspace{2mm}
 \par The smooth loci of $\mathcal{M}(\Sigma)_{\mathbb{C}}$ and $\mathcal{M}(\Sigma)$ have a symplectic form $\omega_{AB}$ defined  by Atiyah and Bott in the context of gauge theory (\cite{AB}). Goldman showed in \cite{Goldman86} that this symplectic structure induces a Poisson bracket on the algebra $\mathcal{O}_{\mathcal{M}_{\mathbb{C}}}$ of regular functions. If $dV$ denotes the Liouville measure associated to this symplectic form, we define the state $\tau^{cl}$ by 
 $$ \tau^{cl}(f):= \frac{1}{\Vol(\mathcal{M}(\Sigma))} \int_{\mathcal{M}(\Sigma)} f dV.$$
 We refer to \cite{FrohmanKania_YMmeasure} for explicit computations of $\tau^{cl}(f_{\gamma})$. 
 We thus have a classical system $\mathcal{A}^{cl}=(\mathcal{O}_{\mathcal{M}(\Sigma)_{\mathbb{C}}}, *, \lVert \cdot \rVert)$ with classical state $\tau^{cl}$.

 \subsection{Mapping class group action and the first Johnson subgroup}

 \par \textbf{Mapping class group } 
 \vspace{2mm}
 \par Let $\Homeo^+(\Sigma)$ denotes the group of orientation-preserving homeomorphisms of $\Sigma$. Being a topological group, its connected component containing the identity map is a normal sub-group and we define the \textit{mapping class group} $\Mod(\Sigma)$ as the quotient of $\Homeo^+(\Sigma)$ by this normal subgroup. It is thus the group of (orientation-preserving) homeomorphisms modulo homotopy (see \cite{FM} for details). In short:
 $$ \Mod(\Sigma):= \pi_0(\Homeo^+(\Sigma)).$$
 \par This group naturally (left) acts on the fundamental group $\pi_1(\Sigma)$ and thus (right) acts on the character variety via the formula 
 $$ \rho^{\phi} (\gamma):= \rho(\phi(\gamma)), \quad \phi \in \Mod(\Sigma), \gamma \in \pi_1(\Sigma), \rho: \pi_1(\Sigma) \to \SL_2(\mathbb{C}).$$. 
 This action is by symplectomorphism so $\tau^{cl}$ is $\Mod(\Sigma)$ invariant.
 \vspace{2mm}
 \par \textbf{Dehn Twists} 
\vspace{2mm}
 \par To any simple closed curve $\gamma\subset \Sigma$ we associate an element $T_{\gamma}\in \Mod(\Sigma)$ defined as follows. Let $S^1\hookrightarrow \Sigma$ be a representing embedding of $\gamma$ and choose $e_{\gamma}:S^1\times [0,1]\hookrightarrow \Sigma$ a tubular neighborhood. We define an orientation-preserving homeomorphism $t_{\gamma}\in \Homeo^+(\Sigma)$ by setting that $t_{\gamma}$ is the identity outside the image of $e_{\gamma}$  and is the image of the homeomorphism $ (s,t)\rightarrow (s\exp(2i\pi t), t)$ of $S^1\times [0,1]$ on the image of $e_{\gamma}$. We call \textit{Dehn twist} along $\gamma$, and denote by $T_{\gamma}\in \Mod(\Sigma)$ the homotopy class of the homeomorphism $t_{\gamma}$. It is a classical result of Lickorish that Dehn twists generate the mapping class group.

\vspace{2mm}
\par \textbf{First Johnson subgroup} 
\vspace{2mm}
\par A curve $\gamma\subset \Sigma$ is called \textit{separating} if $\Sigma \setminus \gamma$ is disconnected. The \textit{first Johnson sub-group} $\mathcal{K}(\Sigma)\subset \Mod(\Sigma)$ is the subgroup generated by Dehn twists along separating curves.

\vspace{2mm}
\par \textbf{Ergodicity}
\vspace{2mm}
\par To apply our quantum ergodic theorem to the character variety, we will use the following results:

\begin{theorem}\label{theorem_mcg_ergodic}
\begin{enumerate}
\item \textbf{\textit{Goldman (\cite{Goldman1})}} The mapping class group acts ergodically on the $\mathrm{SU}(2)$ character variety.
\item \textbf{\textit{Funar-March\'e (\cite{FuMa})}} The first Johnson sub-group acts ergodically on the $\mathrm{SU}(2)$ character variety.
\end{enumerate}
\end{theorem}

\begin{remark} Here we can consider two commutative $C^*$ algebras on which $\Mod(\Sigma)$ acts: the algebra $C(\mathcal{M}(\Sigma))$ of continuous maps $f: \mathcal{M}(\Sigma)\to \mathbb{C}$, where we equip $\mathcal{M}(\Sigma)$ with the analytic topology for which it is compact and Haussdorf and with complex conjugacy,  and the algebra $\mathcal{O}_{\mathcal{M}(\Sigma)_{\mathbb{C}}}$ of regular functions equipped with the $*$-involution defining the real form $\mathcal{M}(\Sigma)$. There is an embedding $i: \mathcal{O}_{\mathcal{M}(\Sigma)_{\mathbb{C}}} \hookrightarrow C(\mathcal{M}(\Sigma))$ of $C^*$ algebras whose image can be thought as the continuous functions which are "polynomials". The image of $i$ is dense in $C(\mathcal{M}(\Sigma))$ for the $\lVert \cdot \rVert$ topology, so $i$ induces an isomorphism $i_* : L^2(\mathcal{O}_{\mathcal{M}(\Sigma)_{\mathbb{C}}}, \tau^{cl}) \xrightarrow{\cong} L^2(C(\mathcal{M}(\Sigma)), \tau^{cl})$. It follows that a group $G\subset \Mod(\Sigma)$ acts ergodically on $C(\mathcal{M}(\Sigma))$ if and only if it acts ergodically on $\mathcal{O}_{\mathcal{M}(\Sigma)_{\mathbb{C}}}$.
\end{remark}

\subsection{Skein algebras and \RT representations}

\par In order to quantize the character variety, we introduce the notion of skein algebras and define some irreducible representations.
\vspace{2mm}
\par \textbf{Skein algebras}
\vspace{2mm}
\par To an oriented compact $3$-manifold $M$, we associate a $\mathbb{Z}[A,A^{-1}]$ module, where $A$ denotes a formal parameter, $\mathcal{S}_A(M)$ defined as follows. The \textit{Kauffman-bracket skein module} $\mathcal{S}_A(M)$ is  the quotient of the free  $\mathbb{Z}[A,A^{-1}]$ module spanned by isotopy classes of framed links embedded in $M$ (including the empty link), by the following skein relations: 
$$
\begin{tikzpicture}[baseline=-0.4ex,scale=0.5,>=stealth]	
\draw [fill=gray!45,gray!45] (-.6,-.6)  rectangle (.6,.6)   ;
\draw[line width=1.2,-] (-0.4,-0.52) -- (.4,.53);
\draw[line width=1.2,-] (0.4,-0.52) -- (0.1,-0.12);
\draw[line width=1.2,-] (-0.1,0.12) -- (-.4,.53);
\end{tikzpicture}
=A
\begin{tikzpicture}[baseline=-0.4ex,scale=0.5,>=stealth] 
\draw [fill=gray!45,gray!45] (-.6,-.6)  rectangle (.6,.6)   ;
\draw[line width=1.2] (-0.4,-0.52) ..controls +(.3,.5).. (-.4,.53);
\draw[line width=1.2] (0.4,-0.52) ..controls +(-.3,.5).. (.4,.53);
\end{tikzpicture}
+A^{-1}
\begin{tikzpicture}[baseline=-0.4ex,scale=0.5,rotate=90]	
\draw [fill=gray!45,gray!45] (-.6,-.6)  rectangle (.6,.6)   ;
\draw[line width=1.2] (-0.4,-0.52) ..controls +(.3,.5).. (-.4,.53);
\draw[line width=1.2] (0.4,-0.52) ..controls +(-.3,.5).. (.4,.53);
\end{tikzpicture}
\hspace{.5cm}
\text{ and }\hspace{.5cm}
\begin{tikzpicture}[baseline=-0.4ex,scale=0.5,rotate=90] 
\draw [fill=gray!45,gray!45] (-.6,-.6)  rectangle (.6,.6)   ;
\draw[line width=1.2,black] (0,0)  circle (.4)   ;
\end{tikzpicture}
= -(A^2+A^{-2}) 
\begin{tikzpicture}[baseline=-0.4ex,scale=0.5,rotate=90] 
\draw [fill=gray!45,gray!45] (-.6,-.6)  rectangle (.6,.6)   ;
\end{tikzpicture}
.$$

\par More precisely, the first relation states that if $L_1,L_2$ and $L_3$ are framed links in $M$ that are identical everywhere except in a small ball in which they  are like respectively the first, second and third pictures, then $[L_1]=A[L_2]+A^{-1}[L_3]$ in the skein module $\mathcal{S}_A(M)$. The second relation states that if $L_2$ is the disjoint union of $L_1$ with a trivial framed unknot, then $[L_2]=-(A^2+A^{-2})[L_1]$ in the skein module. When $M=S^3$, it is easily showed that $\mathcal{S}_A(S^3)$ is one-dimensional, thus $\mathcal{S}_A(S^3)\cong \mathbb{Z}[A,A^{-1}]$. The polynomial associated to such a framed link in $S^3$ is called its Kauffman-bracket polynomial, a framed variant of the celebrated Jones polynomial.
\vspace{2mm}
\par Less obvious is the fact that if $X_g= S^2\times S^1 \# \ldots \# S^2\times S^1$ is $g$ connected sums of $S^2\times S^1$, then $\mathcal{S}_A(X_g)$ is also one-dimensional, we thus get a map $\left<\cdot \right> : \mathcal{S}_A(X_g)\cong \mathbb{Z}[A,A^{-1}]$ sending a framed link $L\subset X_g$ to a polynomial $\left< L \right>$, normalized such that the class of the empty link is $1$.
\vspace{2mm}
\par To a closed oriented connected surface $\Sigma$, we associate the module $\mathcal{S}_A(\Sigma):=\mathcal{S}_A(\Sigma\times [0,1])$. If $\gamma\in \mathcal{MC}$ is a multicurve in the surface, we associate an element $[\gamma]\in \mathcal{S}_A(\Sigma)$ by embedding the curve in $\Sigma\times \{\frac{1}{2}\}\subset \Sigma\times [0,1]$ and associate the parallel framing. We easily show that the image of the set of multicurves forms a basis of $\mathcal{S}_A(\Sigma)$. 
\vspace{2mm}
\par We construct a product on $\mathcal{S}_A(\Sigma)$ turning it into an algebra as follows. If $L_1$ and $L_2$ are two embedded framed links in $\Sigma\times [0,1]$, isotope $L_1$ in $\Sigma\times [0,\frac{1}{2})$ and $L_2$ in $\Sigma\times(\frac{1}{2}, 1] $, then glue $\Sigma\times [0,\frac{1}{2}]$ to $\Sigma\times [\frac{1}{2}, 1]$  to obtain a link $L_1\bigcup L_2 \subset \Sigma \times [0,1]$. The product $[L_1]\cdot [L_2] $ is the class $[L_1\bigcup L_2] \in \mathcal{S}_A(\Sigma)$. 
\vspace{2mm}
\par For $z\in\mathbb{C}$, denote by $\mathcal{S}_{z}(\Sigma)$ the $\mathbb{C}$ algebra $\mathcal{S}_A(\Sigma)\otimes_{A=z}\mathbb{C}$. The relation between the skein algebras and the character variety is given by the following:
\begin{theorem}
\textbf{\textit{Bullock \cite{Bullock}}} The linear map $ \mathcal{S}_{-1}(\Sigma) \rightarrow \mathcal{O}_{\mathcal{M}(\Sigma)_{\mathbb{C}}}$ sending a multicurve $\gamma$ to the regular function $f_{\gamma}$ is an isomorphism of algebras.
\end{theorem}
 More precisely, Bullock proved the above results, assuming that the skein algebra $\mathcal{S}_{-1}(\Sigma)$ was reduced. This fact was latter proved independently in \cite{PS00} and \cite{ChaMa}.
\vspace{2mm}
\par Moreover, the skein algebra produces a star-product in the sense of Section $2$. Consider the element $-\exp\left(\frac{i\pi \hbar}{2}\right)\in \mathbb{C}[[\hbar]]$ and the algebra $\mathcal{S}_{-\exp\left(\frac{i\pi \hbar}{2}\right)}(\Sigma):=\mathcal{S}_A(\Sigma)\otimes_{A=-\exp\left(\frac{i\pi \hbar}{2}\right)} \mathbb{C}[[\hbar]]$. We have a vector space isomorphism
$\mathcal{O}_{\mathcal{M}(\Sigma)}[[\hbar]] \cong \mathcal{S}_{-\exp\left(\frac{i\pi \hbar}{2}\right)}(\Sigma)$ sending $f_{\gamma}$ to $[\gamma]$. We denote by $\star$ the pull-back product. 
Using Goldman's explicit formula for the Poisson bracket $\{f_{\gamma}, f_{\delta}\}$ associated to the Atiyah-Bott symplectic form, Turaev showed:
\begin{theorem} \textbf{\textit{Turaev \cite{Tu}}} We have: 
$$f_{\gamma}\star f_{\delta} = f_{\gamma}f_{\delta} +\circ(\hbar)$$
\par and:
$$  f_{\gamma} \star f_{\delta} - f_{\delta} \star f_{\gamma} = \frac{\hbar}{i}\{f_{\gamma}, f_{\delta}\} +\circ(\hbar^2).$$
\end{theorem}

\par The skein algebras thus produce a deformation quantization of the character variety as defined in Section $2$. Setting $A_N:=-\exp\left(\frac{i\pi}{2N}(N+1)\right)\in \mathbb{C}$, we define $\mathcal{S}_{A_N}(\Sigma):=\mathcal{S}_A(\Sigma)\otimes_{A=A_N}\mathbb{C}$.

\subsection{The \RT representations}

\par In order to get a quantum system, we now construct irreducible representations of $\mathcal{S}_{A_N}(\Sigma)$ following \cite{RT, BHMV2}. A genus $g$ handlebody $H_g$ is a $g$ connected sums of $D^2\times S^1$. Its boundary is a genus $g$ closed oriented surface $\Sigma_g$. By gluing two discs $D^2$ along their boundary with the identity map, we get a sphere $S^2$, thus gluing two handlebodies $H_g^1$ and $H_g^2$ along their boundary with the identity map, we get the manifold $X_g$ which is $g$ connected sums of $S^2\times S^1$. 
\vspace{2mm}\par
Denote by $F_g$ the set of isotopy classes of framed links in $H_g$. We construct a bilinear form $\left<\cdot , \cdot \right>: \mathbb{C}[F_g]\otimes \mathbb{C}[F_g]\rightarrow \mathbb{C}$ as follows. Let $L_1$ and $L_2$ be two framed links embedded in some handlebodies $H_g^1$ and $H_g^2$. By gluing $H_g^1$ and $H_g^2$ along their boundaries with the identity map, we get a link $L_1\bigcup L_2 \subset X_g$. Recall that the skein module of $X_g$ is one dimensional, thus the class $[L_1\bigcup L_2]\in \mathcal{S}_{A_p}(X_g)\cong \mathbb{C}$ is assimilated to a complex number. The pairing $\left<L_1, L_2\right>:=[L_1\bigcup L_2]\in\mathbb{C}$ extends sesquilinearly to give a Hermitian form $\left<\cdot , \cdot \right>: \mathbb{C}[F_g]\otimes \mathbb{C}[F_g]\rightarrow \mathbb{C}$. Since we need a non-degenerate pairing, we define:
$$ V_N(\Sigma_g):= \quotient{ \mathbb{C}[F_g]}{ker \left(\left< \cdot , \cdot \right> \right)}.$$
\par  This space is finite dimensional (see \cite{RT, BHMV2}). We now define a representation
$$ \pi_N: \mathcal{S}_{A_N}(\Sigma_g) \to \End(V_N(\Sigma_g)).$$
\par If $[L]\in V_N(\Sigma_g)$ is a vector given by a framed link in the handlebody $H_g$ and $\gamma\in \mathcal{C}$ is a multicurve in the surface $\Sigma_g$, by identifying the boundary of $H_g$ with $\Sigma_g$ we push $\gamma\in \partial H_g$ inside $H_g$ to get a framed link $\gamma\bigcup L\subset H_g$. We define $[\gamma]\cdot [L]:= [\gamma\bigcup L] $.The action of $\mathcal{S}_{A_N}(\Sigma)$ on $V_N(\Sigma)$ is irreducible (see e.g. \cite{GelcaUribe_SU2}). We thus get a quantization map: 
$$ \Opp_N : \mathcal{O}_{\mathcal{M}(\Sigma)_{\mathbb{C}}} \rightarrow \End(V_N(\Sigma)), \quad \Opp_N(f_{\gamma}):=\pi_N([\gamma]).$$
The following lemma is well-known to the experts though the author failed to find precise references. We postpone its proof to the end of this subsection:
\begin{lemma}\label{lemma_basic}
For our choice of $4N$-root of unity $A_N=-\exp(\frac{i\pi}{2N}(N+1))$, the Hermitian form $\left<\cdot, \cdot\right>$ on $V_N(\Sigma)$ is definite positive. Moreover, $\pi_N$ is a $*$ representation, i.e. the operators $\pi_N([\gamma])$ are self-adjoint.
\end{lemma}

\par Eventually the mapping class group naturally acts on the set of multicurves, inducing an action on the skein algebras and thus on $\End(V_N(\Sigma))$. Since $\pi_N$ is irreducible, this action is by inner automorphisms, thus is we have a (unique) projective representation:
$$ \rho_N: \Mod(\Sigma) \rightarrow \PU(V_N(\Sigma))$$
satisfying the exact Egorov identity. We now take a central extension $\emcg$ of the mapping class and lift the above projective representation to a linear one, the so-called \textit{\RT representation}:
$$\rho_N: \emcg \rightarrow \PU(V_N(\Sigma)).$$

\par As a conclusion, the system $\left( \mathcal{S}_{A_N}(\Sigma), V_N(\Sigma), \Opp_N, \rho_N\right)_{N\geq 3}$ is a quantization of the classical systems $\left(\mathcal{M}(\Sigma), \omega_{AB}, \Mod(\Sigma)\right)$ and $\left(\mathcal{M}(\Sigma), \omega_{AB}, \mathcal{K}(\Sigma)\right)$. Since these groups act ergodically on the character variety, the last ingredient needed to apply Theorem $\ref{main_th}$ is the following result:
\begin{theorem} \textbf{\textit{March\'e-Narrimanejad \cite{MN}}} For any multicurve $\gamma\in \mathcal{MC}$, we have:
\begin{equation}{\label{MN}}
 \frac{1}{\dim(V_N(\Sigma))} \tr \left( \Opp_N(f_{\gamma})\right) \xrightarrow[N\to \infty]{} \frac{1}{\Vol(\mathcal{M}(\Sigma))} \int_{\mathcal{M}} f_{\gamma} dV.
\end{equation}
\end{theorem}

\begin{proof}[Proof of Lemma \ref{lemma_basic}]
We first recall some notations and results of \cite{BHMV2} to which we refer for further details. Write $q:= A^2=-\exp(\frac{i\pi}{N})$, $[n]:= \frac{q^n - q^{-n}}{q-q^{-1}}$, $[n]! = [n][n-1]\ldots [2]$, $I_N:=\{0, \ldots, N-2\}$. A triple $(a,b,c)\in (I_N)^3$ is $N$-admissible if $|a-b|\leq c \leq a+b$ and $a+b+c$ is even and is smaller or equal to $2N-2$. For such a triple, set $i=(a+b-c)/2$, $j=(a+c-b)/2$, $k=(b+c-a)/2$ and write 
$$ \left<a,b,c \right>:= (-1)^{i+j+k} \frac{ [i+j+k+1]! [i]! [j]! [k]!}{[a]! [b]! [c]!}.$$
Let $(\gamma_e)_{e\in E}$ be a pants decomposition of $\Sigma$ (a maximal set of pairwise not homotopic pairwise not intersecting closed curves in $\Sigma$), so that the closures of the connected components of $\Sigma\setminus \cup_e\gamma_e$ are pair of pants $P_v$ (three holed spheres).  Let $\Gamma$ be its dual graph: its vertices $v\in V$ are in $1:1$ correspondence with the pair of pants $P_v$,  its edges $e\in E$ are in $1:1$ correspondence with the curves $\gamma_e$ and the endpoints of $e$ correspond to the pairs of pants adjacent to $\gamma_e$. 
Let $\col_N(\Gamma)$ be the set of maps $\sigma: E \to I_N$ such that if $\gamma_{e_1}\cup \gamma_{e_2}\cup \gamma_{e_3}$ is the boundary of a pair of pants $P_v$ (i.e. $e_1,e_2,e_3$ are adjacent to $v$) then $(\sigma(e_1), \sigma(e_2), \sigma(e_3))$ is $N$-admissible. In this case, we write $\left< \sigma(v)\right>:= \left< \sigma(e_1), \sigma(e_2), \sigma(e_3)\right>$.
To $\sigma \in \col_N(\Gamma)$, one can associate an element $u_{\sigma} \in V_N(\Sigma)$ such that $\{u_{\sigma}, \sigma\in \col_N(\Gamma)\}$ is a basis of $V_N(\Sigma)$. This basis satisfies two important properties: 
\begin{enumerate}
\item one has $\pi_N(\gamma_e) u_{\sigma}= \left< \sigma(e) \right> u_{\sigma} , \quad \mbox{ where } \left<n \right>:= (-1)^n [n+1]$; 
\item  $\{u_{\sigma}, \sigma\in \col_N(\Gamma)\}$ is an orthogonal basis of $(V_N(\Sigma), \left< \cdot, \cdot\right>)$ and 
$$ \left< u_{\sigma}, u_{\sigma} \right>= \frac{ \prod_{v \in V} \left<\sigma(v)\right>}{\prod_{e\in E} \left<\sigma(e)\right>}.$$
\end{enumerate}
To prove that $\left<\cdot , \cdot \right>$ is definite positive, we need to show that $ \left< u_{\sigma}, u_{\sigma} \right> >0$ for all $\sigma \in \col_N(\Gamma)$, so it suffices to proves the following two facts: 
\begin{enumerate}
\item[(i)] for all $n\in I_N$, then $\left<n\right> >0$ and 
\item[(ii)] for all $N$ admissible triple $(a,b,c)$ then $\left<a,b,c\right> >0$.
\end{enumerate}
$(i)$ For $0 \leq n \leq N-2$, then $\left< n \right>= (-1)^n [n+1]= \frac{\sin(i\pi/N)(n+1)}{\sin(i\pi/N)}>0$. $(ii)$ Let $\sgn: \mathbb{R}^* \to \mathbb{Z}/2\mathbb{Z}$ be the group morphism defined by $\sgn(x)\cong 0 \pmod{2}$ if and only if $x>0$. Since $[n]= (-1)^{n+1}\frac{\sin(n\pi/N)}{\sin(\pi/N)}$, we have $\sgn([n])\equiv n+1 \pmod{2}$ so $\sgn( [n]!)\equiv \lfloor \frac{n}{2} \rfloor \pmod{2}$. For $(a,b,c)$ an $N$-admissible triple and $(i,j,k)$ defined as before, we find
$$ \sgn(\left< a,b,c\right>) \equiv \lfloor \frac{i+j+k}{2} \rfloor + \lfloor \frac{i+j}{2} \rfloor+ \lfloor \frac{i+k}{2} \rfloor+ \lfloor \frac{k+j}{2} \rfloor+ \lfloor \frac{i}{2} \rfloor+ \lfloor \frac{j}{2} \rfloor+ \lfloor \frac{k}{2} \rfloor \equiv 0 \pmod{2}.$$
So $\left<a,b,c\right> >0$ and the pairing  $\left<\cdot , \cdot \right>$ is definite positive. It remains to prove that $\pi_N$ is a $*$ representation. Recall that $f_{\gamma}^*=f_{\gamma}$ so we need to show that $\pi_N([\gamma])$ is self-adjoint. This follows either from the definition of $\pi_N$ or from the observation that every non trivial curve $\gamma$ can be extended to a pants decomposition and that  in the orthogonal basis $\{u_{\sigma}, \sigma\in \col_N(\Gamma)\}$, the matrix of $\pi_N([\gamma_e])$ is diagonal with real diagonal entries.

\end{proof}

\subsection{Application to the spin decompositions}

\par We now prove Corollary $\ref{coro_spin}$.  We first briefly recall from \cite{BHMV2} the definition of the spin-decompositions for self-completeness of the paper.
  \vspace{3mm}
  \par Choose $N\geq 8$ such that $4$ divides $N$ and consider the group $G_N \subset \mcg$ generated by the $N$-th power of Dehn twists. In \cite{BHMV2} Proposition $7.5$ and Remark $7.6$, it is shown that the image $\rho_N (G_N)$ is isomorphic to the group $\mathcal{H}$ defined by:
   
  $$ \mathcal{H} := \left\{ \begin{array}{ll}
\mathrm{H}_1 \left( \Sigma_g; \mathbb{Z}/2\mathbb{Z}\right) &\mbox{, if }N \equiv 4 \pmod{8}; \\ 
\mathrm{H}_1 \left( \Sigma_g; \mathbb{Z}/2\mathbb{Z}\right)\rtimes \mathbb{Z}/2\mathbb{Z} &\mbox{, if }N \equiv 0 \pmod{8}. 
\end{array} \right.$$

\par The isomorphism sends $\rho_N(T_{\gamma})^N$ to the class $[\gamma]\in \mathcal{H}$. We thus obtain a decomposition:
$$ V_N(\Sigma) = \bigoplus_{\chi\in \Hom\left(\mathcal{H},\mathbb{Z}/2\mathbb{Z}\right)} V_N\left(\Sigma,\chi\right),$$
\par where: 
  $$ V_N(\Sigma, \chi ) := \left\{ v \in V_N(\Sigma) \mbox{, s.t. } \rho_N(T_{\gamma})^N v = \chi( [\gamma]) v \right\}. $$
  \par It follows from the exact Egorov identity  that $V_N(\Sigma_g, \chi= 0)$ is preserved by $\Mod(\Sigma)$ and that each $V_N(\Sigma, \chi)$ is preserved by $\mathcal{K}(\Sigma)$, since its elements act trivially on homology.
  \vspace{2mm}
  \par Note that, when $8$ divides $N$, then $\Hom\left(\mathcal{H},\mathbb{Z}/2\mathbb{Z}\right)$ is in bijection with the set of spin-structures of $\Sigma$, thus justifying the name \textit{spin-decomposition}. We refer also to \cite{Ma11} for a geometric interpretation of these decompositions.
  
  \vspace{2mm}
  \par To prove Corollary $\ref{coro_spin}$ we want to apply Corollary $\ref{coro_main_th}$. We thus need to control the dimensions of the spaces  $V_N(\Sigma, \chi)$. These dimensions where computed in \cite{BHMV2} from which we deduce the following lemma which permits to conclude:
  \vspace{2mm}
  \begin{lemma}\label{lemma_dim}
  We have the following equivalences:
  \begin{enumerate}
  \item $$ \dim\left(V_{4r}(\Sigma_g)\right) \underset{r\to \infty}{\sim} \Vol (\mathcal{M}(\Sigma_g)) (4r)^{3g-3};$$
  \item $$  \begin{array}{ll}\dim\left(V_{4r}(\Sigma_g, \chi)\right) \underset{r\to \infty}{\sim} 2^{-2g}\Vol (\mathcal{M}(\Sigma_g)) (4r)^{3g-3} &\mbox{, for all } \chi \in \mathcal{H}. \end{array}$$
  \end{enumerate}
  \end{lemma}

  \begin{proof}
  The computation of $\dim(V_N(\Sigma_g))$ was first done by Verlinde in \cite{Verlinde} under the framework of the Wess-Zumino-Witten conformal field theory. An alternative elementary computation was done in \cite[Corollary $1.16$]{BHMV2}  with the use of conformal basis. In particular, it follows from these formulas that $\dim(V_N(\Sigma_g))$ is a polynomial in $N$ of degree $3g-3$. Using March\'e-Narimanejad theorem  of \cite{MN} (Equation $\eqref{MN}$) with $f$ the constant function equal to one, we deduce the first equivalence.
  \vspace{4mm}
  \par In \cite[Theorems $7.10$ and $7.16$]{BHMV2}, the authors computed $\dim\left(V_{4r}(\Sigma_g, \chi)\right)$. When $r$ is odd, they showed that
   $$\dim\left(V_{4r}(\Sigma_g, \chi\neq 0)\right) = 2^{-2g}\left( \dim(V_{4r}(\Sigma_g))-r^{g-1}\right)$$
  and $$\dim\left(V_{4r}(\Sigma_g, \chi= 0)\right) = \dim\left(V_{4r}(\Sigma_g, \chi\neq 0)\right) + r^{g-1}. $$

  \par Moreover when $r$ is even, denoting by $\epsilon \in \{0, 1\}$ the Arf invariant of $\chi$, they showed the formula:
  $$ \dim\left(V_{4r}(\Sigma_g, \chi)\right) = 2^{-2g} \left( \dim(V_{4r}(\Sigma_g))+ r^{g-1} \left( (-1)^{\epsilon} 2^g -1 \right) \right).$$

  \par In every cases, $\dim\left(V_{4r}(\Sigma_g, \chi)\right)$ is a polynomial in $r$ with leading term $2^{-2g}\Vol (\mathcal{M}(\Sigma_g)) (4r)^{3g-3}$. This concludes the proof.
  \end{proof}

\bibliographystyle{alpha}
\bibliography{biblio}

\end{document}